\input ./style/arxiv-vmsta.cfg
\documentclass[numbers,compress,v1.0.1]{vmsta}

\usepackage{dsfont}
\usepackage{vmkcol}

\volume{2}
\issue{4}
\pubyear{2015}
\firstpage{355}
\lastpage{369}
\doi{10.15559/15-VMSTA43}% Updated by VTEXPTS2LaTeX.exe, 14.12.2015
\startlocaldefs

\newcommand{\rrvert}{\vert}
\newcommand{\llvert}{\vert}
\allowdisplaybreaks
\endlocaldefs

\newtheorem{thm}{Theorem}[section]
\newtheorem{lemma}{Lemma}[section]

\newtheorem*{nthm}{Theorem}
\newtheorem*{ntve}{Proposition}
\theoremstyle{remark}
\newtheorem{remark}{Remark}[section]
\theoremstyle{definition}

\newtheorem*{ndefin}{Definition}

\begin{document}
\begin{frontmatter}

\title{Option pricing in the model with stochastic volatility driven by Ornstein--Uhlenbeck process. Simulation}

\author{\inits{S.}\fnm{Sergii}\snm{Kuchuk-Iatsenko}\corref{cor1}}\email
{kuchuk.iatsenko@gmail.com}
\cortext[cor1]{Corresponding author.}
\author{\inits{Yu.}\fnm{Yuliya}\snm{Mishura}}\email{myus@univ.kiev.ua}
\address{Taras Shevchenko National University of Kyiv, Volodymyrska
str. 64, 01601,~Kyiv,~Ukraine}

\markboth{S. Kuchuk-Iatsenko, Yu. Mishura}{Evaluation of the price of European call option}

\begin{abstract}
We consider a discrete-time approximation of paths of an
Ornstein--Uhlenbeck process as a mean for estimation of a price of
European call option in the model of financial market with stochastic
volatility. The Euler--Maruyama approximation scheme is implemented. We
determine the estimates for the option price for predetermined sets of
parameters. The rate of convergence of the price and an average
volatility when discretization intervals tighten are determined.
Discretization precision is analyzed for the case where the exact value
of the price can be derived.
\end{abstract}

\begin{keyword}
Financial markets\sep
stochastic volatility\sep
Ornstein--Uhlenbeck process\sep
option pricing\sep
discrete-time approximations\sep
Euler--Maruyama scheme
\MSC[2010]
91B24\sep
91B25\sep
91G20
\end{keyword}

%\begin{keyword} . \sep.
%\MSC[2010] . \sep.
%\end{keyword}

%
\received{2 December 2015}% Updated by VTEXPTS2LaTeX.exe, 14.12.2015
%12:12
%
\revised{10 December 2015}% Updated by VTEXPTS2LaTeX.exe, 14.12.2015
%12:12
%
\accepted{10 December 2015}% Updated by VTEXPTS2LaTeX.exe, 14.12.2015
%12:12
\publishedonline{17 December 2015}
\end{frontmatter}

%s1 ###
\section{Introduction}
We consider a discrete-time approximation for the price of European
call option in the model of financial market with stochastic volatility
driven by the Ornstein--Uhlenbeck process. An analytic expression for
the price of the option is derived in \cite{paper17}; however, the
resulting formula is complicated and difficult to apply in most of
available software. The discrete-time approximation is ready to be
modeled even in the nonspecific software.

The problem of construction of discrete-time analogues for stochastic
volatility models of financial markets is studied in a series of works
including \cite
{paper20,paper21,paper22,paper23,paper24,paper25,paper26}. Various
techniques are implemented, for example, multilevel Monte Carlo \cite
{paper20}, conditional Monte Carlo \cite{paper22,paper26}, exact
simulation \cite{paper22,paper23}, and It\^o--Taylor approximations
\cite{paper21}.

In most of the works, authors construct discrete-time approximations
both for processes that describe the evolution of the price of asset
and for processes driving the volatility of asset price. The model
considered in this paper allows us to apply another approach: we only
discretize the volatility process. The resulting discrete-time
volatility process is then averaged in a special way and substituted
into the option pricing formula. The option price is determined
conditionally on the path of volatility process, and thus the
conditional Monte Carlo approach is used. The rate of convergence of
the option price calculated using the discrete-time volatility to the
true option price for a given trajectory of volatility process is estimated.

Discretization of the model is naturally connected with the problem of
discrete-time approximations to the solutions of stochastic
differential equations. These matters are widely investigated and
systematized in \cite{book6,book7,book8}. The simplest discrete-time
approximation is the stochastic generalization of Euler approximation
for deterministic differential equations proposed in \cite{paper18},
which is also referred to as the Euler--Maruyama scheme. Another
suitable for implementation and effective method is the Milstein scheme
\cite{paper19}. Since the model under consideration is a diffusion with
additive noise, both schemes coincide which is referred to below. It is
worth noticing that Euler and Milstein schemes both belong to the class
of It\^o--Taylor approximations and have orders of convergence 0.5 and
1, respectively. For some diffusions, the approximation schemes can be
enhanced to provide higher-order convergence, but this usually results
in great increase in computation time.

Although exact simulation provides more precision compared to the Euler
approximation, in this paper, we use the latter. This is motivated by
the fact that the Euler approximation is cheaper in terms of
computation time and by our desire to assess the rate of convergence of
conditional option prices when the volatility is discretized using the
Euler scheme.

This paper is structured as follows. We begin with the definition of
the model under consideration and the discretization scheme used. In
Section~\ref{sec3}, the prices of the European call option are compared for
discrete-time and continuous volatility processes to derive the
estimate of strong convergence order. Section~\ref{sec4} provides numeric
results of the simulation. In Section~\ref{sec5}, we demonstrate the precision
of discrete-time approximation for the case of deterministic
volatility. Appendix A contains definitions and auxiliary results on
discretization schemes and orders of their convergence mostly coming
from \cite{book6}.

%s2 ###
\section{The model and discrete approximation of volatility process}
\label{sec2}
Let $\{\varOmega, \mathcal{F}, \mathbf{F}=\{\mathcal{F}_t^{(B,Z)},
t\geq0\}, \mathbb{P}\}$ be a complete probability space with
filtration generated by Wiener processes $\{B_t$, $Z_t$, $0 \leq t \leq
T\}$. We consider the model of the market where one risky asset is
traded, its price evolves according to the geometric Brownian motion $\{
S_t,\;0 \leq t \leq T\}, $ and its volatility is driven by a stochastic
process. More precisely, the market is described by the pair of
stochastic differential equations

%e1 ###
\begin{equation}
\label{ModelA0} dS_t = \mu S_tdt+\sigma(Y_t)S_tdB_t,
\end{equation}
%
%e2 ###
\begin{equation}
\label{ModelA1} dY_t = -\alpha Y_tdt+kdZ_t.
\end{equation}

We denote by $S_0=S $ and $Y_0=Y$ the deterministic initial values of
the processes specified by Eqs.~\eqref{ModelA0}--\eqref{ModelA1}.

We further impose the following assumptions:

\begin{itemize}

\item[(C1)] The Wiener processes $B$ and $Z$ are uncorrelated;

\item[(C2)] the volatility function $\sigma: \mathbb{R}\rightarrow
\mathbb{R}_+$ is measurable, bounded away from zero by a constant $c$:
\[
\sigma(x) \geq c > 0, \quad x \in\mathds{R},
\]
and satisfies the condition $\int_0^T\sigma^2(Y_t)dt<\infty$ a.s.;

\item[(C3)] the coefficients $\alpha$, $\mu$, and $k$ are positive.
\end{itemize}

For example, the conditions mentioned in assumption (C2) are satisfied
for the measurable function $\sigma(x)$ such that $c \leq\sigma^2(x)
\leq C$ for $0 < x < T$ and some constants $0<c<C$. Moreover, given the
square integrability of $\sigma(Y_s)$, the solution of differential
equation \eqref{ModelA0} is given by
%
%e3 ###
\begin{equation}
S_t=S_0\exp{ \Biggl(\mu t - \frac{1}{2}\int
^t_0\sigma^2(Y_s)ds+
\int^t_0\sigma(Y_s)dB_s
\Biggr)},
\end{equation}
which yields that $S_t$ is continuous. Hence, the product $\sigma
(Y_s)S_t$ is square integrable: $\int_0^T\sigma^2(Y_t)S^2_tdt<\infty$ a.s.

The unique solution of the Langevin equation \eqref{ModelA1} $Y_t$ is
the so called Ornstein--Uhlen\-beck (OU) process. Its properties make
it a suitable tool for modeling volatility in financial markets. One of
the most important of the features is the mean-reversion property. The
OU process is Gaussian with the following characteristics:
\begin{eqnarray*}
E[Y_t]=Y_0\operatorname{e}^{-\alpha t}, \qquad
\operatorname{Var}[Y_t]=\frac{k^2}{2\alpha} \bigl(1-\operatorname
{e}^{-2\alpha t} \bigr).
\end{eqnarray*}

Moreover, the OU process is Markov and admits the explicit representation

%e4 ###
\begin{equation}
\label{OU} Y_t=Y_0 \operatorname{e}^{-\alpha t}+k
\int_{0}^t \operatorname{e}^{-\alpha(t-s)}dZ_s.
\end{equation}

Following \cite{paper17}, we proceed to the risk-neutral setting
characterized by the minimal martingale measure $\mathbb{Q}$. With $r$
being the interest rate, Eqs.~\eqref{ModelA0}--\eqref{ModelA1} are now
in the following form (see Section 5 in \cite{paper17}):
%
%e5 ###
\begin{align}
\label{ModelB} %
dS_t & =  r S_tdt+\sigma(Y_t)S_tdB^\mathbb{Q}_t, \nonumber\\
dY_t & = -\alpha Y_tdt+kdZ^\mathbb{Q}_t,
\end{align}
where
\begin{align*}
B^\mathbb{Q}_t & = B_t+\int_0^t
\dfrac{\mu-r}{\sigma(Y_s)}ds\quad\mbox{and}\quad Z^\mathbb{Q}_t =
Z_t
\end{align*}
are independent Wiener processes w.r.t.\ $\mathbb{Q}$.

This continuous-time model admits a variety of discrete-time
approximations. In this paper, we apply the familiar Euler--Maruyama
scheme, also referred to as the Euler scheme. The Euler--Maruyama
approximation to the true solution of the Langevin equation \eqref
{ModelA1} is the Markov chain $Y^{(m)}$ defined as follows:
\begin{itemize}
\item the partition of the interval $[0,T]$ into $m$ equal subintervals
of width $\Delta t=T/m$ is considered;
\item the initial value of the scheme is set: $Y^{(m)}_0=Y_0$;
\item$Y^{(m)}_{l+1}$, which we will use as a shorthand for
$Y^{(m)}_{(l+1)T/m}$, $0 \leq l \leq m-1$, is recursively defined by
%
%e6 ###
\begin{equation}
\label{Yml} Y^{(m)}_{l+1}=(1-\alpha\Delta
t)Y^{(m)}_{l}+k \Delta Z^{\mathbb{Q}}_l,
\end{equation}
where $\Delta Z^{\mathbb{Q}}_l= Z^{\mathbb{Q}}_{(l+1)T/m}-Z^{\mathbb
{Q}}_{lT/m}$.
\end{itemize}
The continuous-time process $Y^{(m)}_t$ is a step-type process defined by
\begin{equation*}
Y^{(m)}_t=Y^{(m)}_{[tm/T]T/m}, \quad t \in[0,T],
\end{equation*}
where $[x]$ denotes an integer part of $x$.

%s3 ###
\section{The price of European call option}\label{sec3}

The price of European call option $V$ in the initial time moment of in
model~\eqref{ModelB} is provided by
%
%e7 ###
\begin{equation}
\label{EV_0} V=\operatorname{e}^{-rT}\mathbb{E}^{\mathbb{Q}} \bigl\{
\mathbb{E}^{\mathbb{Q}} \bigl\{ \bigl(S^{\mathbb{Q}}_T-K
\bigr)^{+}\,\big|\,Y_s, 0 \leq s\leq T \bigr\} \bigr\}.
\end{equation}

The inner expectation is conditional on the path of $Y_s,\ 0 \leq s
\leq T$, and therefore, it actually is the Black--Scholes price for a
model with deterministic time-dependent volatility. According to Lemma
2.1 in \cite{paper4}, the inner expectation in \eqref{EV_0} has the
following representation:
\begin{align}
\notag P&:=\mathbb{E}^{\mathbb{Q}} \bigl\{ \bigl(S^{\mathbb{Q}}_T-K
\bigr)^{+}\,\big|\,Y_s, 0 \leq s\leq T \bigr\}=
\operatorname{e}^{\ln{S}+rT}\varPhi(d_1)-K\varPhi(d_2)
\\
&:=\operatorname{e}^{\ln{S}+rT}\varPhi \biggl(\frac{\ln S+(r+\frac{1}{2}\bar{\sigma}^2)T-\ln K}{\bar{\sigma}\sqrt{T}} \biggr)\notag
\\
&\quad -K\varPhi \biggl(\frac{\ln S+(r-\frac{1}{2}\bar{\sigma}^2)T-\ln K}{\bar{\sigma}\sqrt{T}} \biggr),\label{IntExp}
\end{align}
where $\bar{\sigma}:=\sqrt{\frac{1}{T}\int_0^T\sigma^2(Y_s)ds} \geq
0$, $\varPhi$ is the standard normal distribution function. The
function $\bar{\sigma}$ may be viewed as the volatility averaged from
the initial moment of time to maturity. The arguments of $\varPhi$ in
\eqref{IntExp} are denoted as $d_1$ and~$d_2$.

Our aim is to estimate the error arising as a result of approximation
of the exact formula \eqref{EV_0} by application of the Euler
approximation to the process that drives volatility. Thus, we need to
assess the expectation of $R$ given by
%
%e8 ###
\begin{equation}
R:=|P-\hat{P}_{m} |,
\end{equation}
where\ $m$ is the number of discretization points dividing the time
interval $[0,T]$ into equal intervals, $\hat{P}_{m}$ denotes the price
of the option in discrete setting calculated using the formula similar
to \eqref{IntExp}:
%
%e9 ###
\begin{equation}
\label{V0discr} \hat{P}_{m} =\operatorname{e}^{\ln{S}+rT}\varPhi
\bigl(d^{(m)}_1 \bigr)-K\varPhi \bigl(d^{(m)}_2
\bigr),
\end{equation}
where
\begin{gather}
d^{(m)}_1=\frac{\ln S+(r+\frac{1}{2}\bar{\sigma}^2_{m})T-\ln K}{\bar
{\sigma}_{m}\sqrt{T}},
\\
d^{(m)}_2=\frac{\ln S+(r-\frac{1}{2}\bar{\sigma}^2_{m})T-\ln K}{\bar
{\sigma}_{m}\sqrt{T}},
\end{gather}
with\
%
%e10 ###
\begin{equation}
\label{sigmadiscr} \bar{\sigma}_{m}=\sqrt{
\frac{1}{T}\sum_{l=1}^{m}
\sigma^2 \bigl(Y^{(m)}_l \bigr)
\frac{T}{m}}=\sqrt{\frac{1}{m}\sum
_{l=1}^{m}\sigma^2 \bigl(Y^{(m)}_l
\bigr)},
\\
\end{equation}
where $Y^{(m)}_l$ is defined in \eqref{Yml}.

It is unlikely that we are able to find an exact or even approximate
value for~$R$. However, what really makes interest for investigation of
the above bundle of models is the rate of convergence of the discrete
setting to the continuous one. In order to assess the rate of
convergence, the expression for an upper bound of $R$ in terms of $m$
needs to be derived.

Comparing \eqref{IntExp} and \eqref{V0discr}, we can see that the
approximation error arises solely due to the difference between $\bar
{\sigma}$ and $\bar{\sigma}_{m}$. So, the first step would assessing
the upper bound of expectation of absolute value of this difference
w.r.t.~$m$. After that, $R$ might be expressed in terms of $R_{\sigma
}:=\mathbb{E}|\bar{\sigma}-\bar{\sigma}_{m}|$.

\begin{lemma}\label{LemSigma}
Let $\sigma^2(x)$ satisfy the H\"{o}lder condition
%
%e11 ###
\begin{equation}
\label{Holder}
\big|\sigma^2(x)-\sigma^2(y)\big|\leq L|x-y|^{\gamma},
\end{equation}
where $0< \gamma\leq1$, and $L$ is some positive constant.
Then $\mathbb{E}R_{\sigma} \leq Cm^{-0.5\gamma}$, where $C$ is some
positive constant.
\end{lemma}
\begin{proof}
Since $\bar{\sigma}_{m}$ and $\bar{\sigma}$ are both square root
functions, it is be more convenient to work with $\bar{\sigma}^2_{m}$
and $\bar{\sigma}^2$. To this end, we will use H\"{o}lder's inequality:
\begin{align*}
\mathbb{E}|\bar{\sigma}_{m}-\bar{\sigma}| &= \mathbb{E} \Biggl\llvert
\sqrt{\dfrac{1}{T}\int_0^T
\sigma^2(Y_s)ds}-\sqrt{
\frac{1}{m}\sum_{i=1}^{m}
\sigma^2 \bigl(Y^{(m)}_i \bigr)} \Biggr\rrvert
\\
&\leq\mathbb{E} \Biggl\llvert\sqrt{ \Biggl\llvert
\dfrac{1}{T}\int_0^T\sigma^2(Y_s)ds-
\frac{1}{m}\sum_{i=1}^{m}
\sigma^2 \bigl(Y^{(m)}_i \bigr) \Biggr\rrvert}
\Biggr\rrvert
\\
&\leq \Biggl(\mathbb{E} \Biggl\llvert\dfrac{1}{T}\int_0^T
\sigma^2(Y_s)ds-\frac{1}{m}\sum
_{i=1}^{m}\sigma^2 \bigl(Y^{(m)}_i
\bigr) \Biggr\rrvert \Biggr)^{1/2}.
\end{align*}
Now we represent the integral as a sum of integrals over shorter
intervals. Since the second summand does not depend on $s$, we may move
it inside the integral sign, multiplying it by the inverse to the
interval length:
\begin{align*}
\mathbb{E}|\bar{\sigma}_{m}-\bar{\sigma}| &\,{\leq}\, \Biggl(
\mathbb{E} \Biggl\llvert\sum_{i=0}^{m-1}
\Biggl(\dfrac{1}{T}\int_{iT/m}^{(i+1)T/m}
\sigma^2(Y_s)ds-\frac{1}{m}\sigma^2
\bigl(Y^{(m)}_{i+1} \bigr) \Biggr) \Biggr\rrvert
\Biggr)^{1/2}
\\
&\,{=}\, \Biggl(\mathbb{E} \Biggl\llvert\sum_{i=0}^{m-1}
\Biggl(\dfrac{1}{T}\int_{iT/m}^{(i+1)T/m}\!
\sigma^2(Y_s)ds\,{-}\,\frac{1}{m}
\frac{m}{T}\int_{iT/m}^{(i+1)T/m}\!
\sigma^2 \bigl(Y^{(m)}_{i+1} \bigr)ds \Biggr) \Biggr
\rrvert \Biggr)^{1/2}
\\
&\,{=}\, \Biggl(\mathbb{E} \Biggl\llvert\dfrac{1}{T}\sum
_{i=0}^{m-1}\int_{iT/m}^{(i+1)T/m}
\bigl(\sigma^2(Y_s)-\sigma^2
\bigl(Y^{(m)}_{i+1} \bigr) \bigr)ds \Biggr\rrvert
\Biggr)^{1/2}.
\end{align*}
We apply the H\"{o}lder property of $\sigma^2(x)$:
\begin{align*}
& \Biggl(\mathbb{E} \Biggl\llvert\dfrac{1}{T}\sum
_{i=0}^{m-1}\int_{iT/m}^{(i+1)T/m}
\bigl(\sigma^2(Y_s)-\sigma^2
\bigl(Y^{(m)}_{i+1} \bigr) \bigr)ds \Biggr\rrvert
\Biggr)^{1/2}
\\
&\quad \leq \Biggl(\dfrac{L}{T}\mathbb{E} \Biggl(\sum
_{i=0}^{m-1}\int_{iT/m}^{(i+1)T/m}
\bigl\llvert Y_s - Y^{(m)}_{i+1} \bigr
\rrvert^{\gamma}ds \Biggr) \Biggr)^{1/2}
\\
&\quad = \Biggl(\dfrac{L}{T}\sum_{i=0}^{m-1}
\int_{iT/m}^{(i+1)T/m}\mathbb{E} \bigl\llvert
Y_s- Y^{(m)}_{i+1} \bigr\rrvert^{\gamma}ds
\Biggr)^{1/2}.
\end{align*}

Recall that $ Y^{(m)}_i$ is a shorthand for $Y^{(m)}_{iT/m}=Y^{(m)}_s$,
$s \in[iT/m, (i+1)T/m)$, and Proposition from the Appendix A yields
that $\mathbb{E}\llvert Y_s - Y^{(m)}_{i+1}\rrvert \leq C_1m^{-1}$,
where $C_1$ is some positive constant. We use H\"{o}lder's inequality
to derive that $\mathbb{E}\llvert Y_s - Y^{(m)}_{i+1}\rrvert^{\gamma}
\leq C_1^{\gamma} m^{-\gamma}$ and arrive at
\begin{equation*}
\mathbb{E}|\bar{\sigma}_{m}-\bar{\sigma}| \leq \biggl(
\dfrac{L}{T}m\frac{T}{m}C_1^{\gamma}m^{-\gamma}
\biggr)^{1/2}=Cm^{-\gamma/2}
\end{equation*}
for $C:=\sqrt{L C_1^{\gamma}}$, which proves the lemma.
\end{proof}

The above lemma enables us to prove the main result of this section.
\begin{thm}
Let $\sigma^2(x)$ satisfy H\"{o}lder condition \eqref{Holder}. Then
$\mathbb{E}R \leq D m^{-\gamma/2}$, where $D$ is some positive constant.
\end{thm}
\begin{proof}
The function $\varPhi(x)$ has a continuous bounded derivative on
$\mathds{R}$; hence, we can use its Lipschitz property:
\begin{align*}
\mathbb{E}R&=\mathbb{E}\llvert P-\hat{P}_{m}\rrvert
\\
&\leq\mathbb{E} \bigl(S\operatorname{e}^{rT} \bigl\llvert
\varPhi(d_1)-\varPhi \bigl(d^{(m)}_1 \bigr)
\bigr\rrvert+K \bigl\llvert\varPhi(d_2)-\varPhi \bigl(d^{(m)}_2
\bigr) \bigr\rrvert \bigr)
\\
&\leq\mathbb{E} \Bigl(\sup_x \big|f(x)\big| \bigl(S
\operatorname{e}^{rT}\big|d_1-d^{(m)}_1\big|+K\big|d_2-d^{(m)}_2\big|
\bigr) \Bigr),
\end{align*}
where $f(x)=\frac{1}{\sqrt{2\pi}}\operatorname{e}^{-\frac{x^2}{2}}$ is
the density of the standard normal distribution.
In the above representation,
\begin{align}
\big|d_1-d^{(m)}_1\big|&= \biggl\llvert \biggl(
\frac{1}{\bar{\sigma}}-\frac{1}{\bar{\sigma}_{m}} \biggr)\frac{\ln S-\ln K + rT}{\sqrt{T}}+\frac{1}{2}
\sqrt{T}(\bar{\sigma}-\bar{\sigma}_{m}) \biggr\rrvert
\\
&\leq|\bar{\sigma}-\bar{\sigma}_{m}| \biggl\llvert\frac{1}{\bar{\sigma}\bar{\sigma}_{m}}
\frac{\ln(S/K) + rT}{\sqrt{T}}+\frac{\sqrt{T}}{2} \biggr\rrvert
\\
&\leq|\bar{\sigma}-\bar{\sigma}_{m}| \biggl\llvert\frac{\ln(S/K) + rT}{c^2\sqrt{T}}+
\frac{\sqrt{T}}{2} \biggr\rrvert,
\end{align}
where $c$ is a positive constant, and the last inequality is due to the
assumption that $\sigma(x)$ is bounded away from zero for any $x \in
\mathds{R}$ (see assumption (C2)). Hence, using Lemma~\ref{LemSigma},
we get
\begin{equation*}
\mathbb{E}\big|d_1-d^{(m)}_1\big| \leq
C_1 \mathbb{E}|\bar{\sigma}-\bar{\sigma}_{m}| \leq
D_1 m^{-\gamma/2},
\end{equation*}
where $C_1:=\llvert\frac{\ln(S/K) + rT}{c^2\sqrt{T}}+\frac{\sqrt
{T}}{2}\rrvert$ and $D_1$ are positive constants.

Similarly, $\mathbb{E}|d_2-d^{(m)}_2| \leq D_2 m^{-\gamma/2}$,
$D_2=const>0$, and we arrive at
%
%e12 ###
\begin{equation}
\mathbb{E}R=\mathbb{E}\llvert P-\hat{P}_{m}\rrvert\leq
\frac{1}{\sqrt{2\pi}} \bigl(D_1 S\operatorname{e}^{rT}m^{-\gamma/2}+D_2
Km^{-\gamma/2} \bigr)=Dm^{-\gamma/2}
\end{equation}
for a positive constant $D$.

The theorem is proved.
\end{proof}

%s4 ###
\section{Numeric examples}\label{sec4}

Theorem 6.1 in \cite{paper17} provides an analytic representation for
the price of European call option for the stochastic volatility model
under consideration. However, using it to calculate the price of an
option is rather difficult and time-consuming. We further present the
results of calculation of the price of European call option using
simulation techniques.

The calculation process is performed in Matlab 7.9.0 and is structured
as follows:
\begin{itemize}
\item[1.] The choice of discrete ranges of values of input parameters;
\item[2.] The choice of the function $\sigma(Y_s)$;
\item[3.] For each combination of input \xch{parameters}{parametersm} we generate 1000
trajectories of an Ornstein--Uhlenbeck process by splitting the time
interval into subintervals of length $\Delta t=0.001$ and modeling
values of the OU process at these points (that is, generating normally
distributed variables with known mean and standard deviation using
relationship \eqref{Yml}). For each trajectory, \eqref{V0discr} is
applied to calculate $\bar{\sigma}^2_{m}$ and the price of an option.
The results for all trajectories are then averaged and discounted to
provide the sample average of the price denoted by $\hat{\mathbb{E}}\hat
{P}_{m}$. The average volatility over all trajectories and time
interval is denoted by $\hat{\mathbb{E}}\bar{\sigma}^2_{m}$.
\end{itemize}

% Table generated by Excel2LaTeX from sheet 'ax+b,1'
%
%t1 ###
\begin{table}[b!]
\tabcolsep=0pt
\centering
\caption{$\sigma^2(Y_s)=a|Y_s|+b$}
\label{tab1}
\begin{tabular*}{\textwidth}{@{\extracolsep{\fill}}D{.}{.}{1.2}D{.}{.}{1.1}D{.}{.}{1.2}D{.}{.}{1.1}cD{.}{.}{1.1}cccc@{}}
\hline\\[-8pt]
\multicolumn{1}{@{}c}{$T$} & \multicolumn{1}{c}{$k$} & \multicolumn{1}{c}{$r$} & \multicolumn{1}{c}{$K$} & \multicolumn{1}{c}{$a$} & \multicolumn{1}{c}{$b$} & \multicolumn{1}{c}{$\hat{\mathbb{E}}\bar{\sigma}^2_{m}$} & \multicolumn{1}{c}{$\hat{\mathbb{E}}\hat{P}_{m}$} & $\hat{\mathbb{E}}\bar{\sigma}^2_{m}$ & {$\hat{\mathbb{E}}\hat{P}_{m}$} \\
\hline
& & & & & & \multicolumn{2}{@{}c@{}}{$\alpha=1$} & \multicolumn{2}{c}{$\alpha=100$} \\
\cline{7-8}\cline{9-10}
0.25 & 0.1 & 0    & 0.8 & 1 & 0   & 0.088 & \textbf{0.204} & 0.009 & \textbf{0.200}  \\
0.5  & 0.1 & 0    & 0.8 & 1 & 0   & 0.082 & \textbf{0.213} & 0.007 & \textbf {0.200} \\
1    & 0.1 & 0    & 0.8 & 1 & 0   & 0.073 & \textbf{0.227} & 0.007 & \textbf {0.200} \\
0.25 & 0.5 & 0    & 0.8 & 1 & 0   & 0.147 & \textbf{0.211} & 0.031 & \textbf {0.200} \\
0.5  & 0.5 & 0    & 0.8 & 1 & 0   & 0.185 & \textbf{0.235} & 0.030 & \textbf {0.201} \\
1    & 0.5 & 0    & 0.8 & 1 & 0   & 0.216 & \textbf{0.280} & 0.029 & \textbf {0.207} \\
0.25 & 1   & 0    & 0.8 & 1 & 0   & 0.264 & \textbf{0.224} & 0.059 & \textbf {0.201} \\
0.5  & 1   & 0    & 0.8 & 1 & 0   & 0.338 & \textbf{0.264} & 0.058 & \textbf {0.207} \\
1    & 1   & 0    & 0.8 & 1 & 0   & 0.412 & \textbf{0.334} & 0.058 & \textbf {0.221} \\
0.25 & 0.1 & 0.01 & 1   & 1 & 0.2 & 0.289 & \textbf{0.108} & 0.209 & \textbf{0.092}  \\
0.5  & 0.1 & 0.01 & 1   & 1 & 0.2 & 0.281 & \textbf{0.151} & 0.207 & \textbf{0.130}  \\
1    & 0.1 & 0.01 & 1   & 1 & 0.2 & 0.273 & \textbf{0.210} & 0.207 & \textbf {0.184} \\
0.25 & 0.5 & 0.01 & 1   & 1 & 0.2 & 0.346 & \textbf{0.117} & 0.231 & \textbf{0.097}  \\
0.5  & 0.5 & 0.01 & 1   & 1 & 0.2 & 0.375 & \textbf{0.172} & 0.230 & \textbf{0.137}  \\
1    & 0.5 & 0.01 & 1   & 1 & 0.2 & 0.414 & \textbf{0.254} & 0.229 & \textbf {0.193} \\
0.25 & 1   & 0.01 & 1   & 1 & 0.2 & 0.459 & \textbf{0.134} & 0.259 & \textbf{0.102}  \\
0.5  & 1   & 0.01 & 1   & 1 & 0.2 & 0.532 & \textbf{0.203} & 0.258 & \textbf {0.145} \\
1    & 1   & 0.01 & 1   & 1 & 0.2 & 0.617 & \textbf{0.305} & 0.258 & \textbf {0.204} \\
0.25 & 0.1 & 0.02 & 1.2 & 1 & 1   & 1.089 & \textbf{0.141} & 1.009 & \textbf{0.134}  \\
0.5  & 0.1 & 0.02 & 1.2 & 1 & 1   & 1.079 & \textbf{0.228} & 1.007 & \textbf{0.218}  \\
1    & 0.1 & 0.02 & 1.2 & 1 & 1   & 1.073 & \textbf{0.347} & 1.007 & \textbf {0.335} \\
0.25 & 0.5 & 0.02 & 1.2 & 1 & 1   & 1.148 & \textbf{0.147} & 1.031 & \textbf{0.136}  \\
0.5  & 0.5 & 0.02 & 1.2 & 1 & 1   & 1.178 & \textbf{0.240} & 1.030 & \textbf{0.221}  \\
1    & 0.5 & 0.02 & 1.2 & 1 & 1   & 1.216 & \textbf{0.371} & 1.029 & \textbf {0.339} \\
0.25 & 1   & 0.02 & 1.2 & 1 & 1   & 1.262 & \textbf{0.157} & 1.059 & \textbf{0.138}  \\
0.5  & 1   & 0.02 & 1.2 & 1 & 1   & 1.341 & \textbf{0.260} & 1.058 & \textbf {0.225} \\
1    & 1   & 0.02 & 1.2 & 1 & 1   & 1.414 & \textbf{0.402} & 1.058 & \textbf {0.344} \\
\hline
\end{tabular*}
\end{table}

To begin with, let us recall the notation of input parameters along
with ranges of values assigned to them in the process of simulation:

$T$ -- time to maturity, $T=0.25; \ 0.5; \ 1$;

$k$ -- volatility of OU process, $k=0.1; \ 0.5; \ 1$;

$\alpha$ -- mean-reversion rate, $\alpha=1; \ 100$;

$r$ -- interest rate, $r= 0; \ 0.01; \ 0.02$;

$K$ -- strike price, $K= 0.8; \ 1; \ 1.2$;

$S_0$ -- initial price of stock, $S_0=1$;

$Y_0$ -- initial value of OU process, $Y_0=0.1$.

In order to produce numerical results, we choose the following options
for the function $\sigma(Y_s)$:
\begin{itemize}
\item[1.] $\sigma^2(Y_s)=a|Y_s|+b$, where $a=\{0,1\}$, $b=\{0,0.2,1\}$ \xch{(Table~\ref{tab1});}{;}
\item[2.] $\sigma^2(Y_s)=\operatorname{e}^{Y_s}+c$, $c=0.02$ \xch{(Table~\ref{tab2}).}{.}
\end{itemize}

The results of simulations are split into groups by the mean-reversion
rate $\alpha$ and function $\sigma(Y_s)$. Meaningless and uninteresting
results provided by some distinct combinations of inputs are ignored.

Mean-reversion of 1 corresponds to slow reverting models, and fast
mean-rever\-ting models are characterized by $\alpha=100$. Matters of
speed of mean-reversion are addressed, for example, in \cite{book1}.

% Table generated by Excel2LaTeX from sheet 'exp,1'
%
%t2 ###
\begin{table}[t]
\tabcolsep=0pt
\caption{$\sigma^2(Y_s)=\operatorname{e}^{Y_s}+c$}
\label{tab2}
\begin{tabular*}{\textwidth}{@{\extracolsep{\fill}}D{.}{.}{1.2}D{.}{.}{1.1}D{.}{.}{1.2}D{.}{.}{1.1}cccc@{}}
\hline\\[-8pt]
\multicolumn{1}{c}{$T$} & \multicolumn{1}{c}{$k$} & \multicolumn{1}{c}{ $r$} & \multicolumn{1}{c}{$K$ } & $\hat{\mathbb{E}}\bar{\sigma}^2_{m}$ & {$\hat{\mathbb{E}}\hat{P}_{m}$} & $\hat{\mathbb{E}}\bar{\sigma}^2_{m}$ & {$\hat{\mathbb{E}}\hat{P}_{m}$} \\
\hline
                         &                         &                           &                         & \multicolumn{2}{c}{$\alpha=1$}                  & \multicolumn{2}{c}{$\alpha=100$}                   \\\cline{5-6}\cline{7-8}
0.25 & 0.1 & 0    & 0.8 & 1.113                          & \textbf{0.303} & 1.024                            & \textbf{0.297}  \\
0.5  & 0.1 & 0    & 0.8 & 1.103                          & \textbf{0.372} & 1.022                            & \textbf{0.363}  \\
1    & 0.1 & 0    & 0.8 & 1.088                          & \textbf{0.465} & 1.021                            & \textbf{0.456}  \\
0.25 & 0.5 & 0    & 0.8 & 1.135                          & \textbf{0.305} & 1.025                            & \textbf{0.297}  \\
0.5  & 0.5 & 0    & 0.8 & 1.131                          & \textbf{0.374} & 1.023                            & \textbf{0.363}  \\
1    & 0.5 & 0    & 0.8 & 1.119                          & \textbf{0.468} & 1.022                            & \textbf{0.456}  \\
0.25 & 1   & 0    & 0.8 & 1.184                          & \textbf{0.307} & 1.027                            & \textbf{0.297}  \\
0.5  & 1   & 0    & 0.8 & 1.212                          & \textbf{0.380} & 1.025                            & \textbf{0.363}  \\
1    & 1   & 0    & 0.8 & 1.238                          & \textbf{0.478} & 1.024                            & \textbf {0.456} \\
0.25 & 0.1 & 0.01 & 1   & 1.112                          & \textbf{0.209} & 1.024                            & \textbf{0.201}  \\
0.5  & 0.1 & 0.01 & 1   & 1.103                          & \textbf{0.291} & 1.022                            & \textbf{0.281}  \\
1    & 0.1 & 0.01 & 1   & 1.086                          & \textbf{0.401} & 1.021                            & \textbf{0.390}  \\
0.25 & 0.5 & 0.01 & 1   & 1.121                          & \textbf{0.209} & 1.025                            & \textbf{0.201}  \\
0.5  & 0.5 & 0.01 & 1   & 1.129                          & \textbf{0.294} & 1.023                            & \textbf{0.281}  \\
1    & 0.5 & 0.01 & 1   & 1.128                          & \textbf{0.405} & 1.022                            & \textbf{0.390}  \\
0.25 & 1   & 0.01 & 1   & 1.178                          & \textbf{0.213} & 1.026                            & \textbf{0.201}  \\
0.5  & 1   & 0.01 & 1   & 1.206                          & \textbf{0.299} & 1.025                            & \textbf{0.281}  \\
1    & 1   & 0.01 & 1   & 1.216                          & \textbf{0.412} & 1.023                            & \textbf{0.390}  \\
0.25 & 0.1 & 0.02 & 1.2 & 1.110                          & \textbf{0.143} & 1.024                            & \textbf{0.135}  \\
0.5  & 0.1 & 0.02 & 1.2 & 1.103                          & \textbf{0.231} & 1.022                            & \textbf{0.220}  \\
1    & 0.1 & 0.02 & 1.2 & 1.087                          & \textbf{0.349} & 1.021                            & \textbf{0.338}  \\
0.25 & 0.5 & 0.02 & 1.2 & 1.133                          & \textbf{0.145} & 1.025                            & \textbf{0.135}  \\
0.5  & 0.5 & 0.02 & 1.2 & 1.128                          & \textbf{0.233} & 1.023                            & \textbf{0.220}  \\
1    & 0.5 & 0.02 & 1.2 & 1.115                          & \textbf{0.352} & 1.021                            & \textbf{0.338}  \\
0.25 & 1   & 0.02 & 1.2 & 1.162                          & \textbf{0.147} & 1.027                            & \textbf{0.135}  \\
0.5  & 1   & 0.02 & 1.2 & 1.201                          & \textbf{0.239} & 1.025                            & \textbf{0.220}  \\
1    & 1   & 0.02 & 1.2 & 1.255                          & \textbf{0.367} & 1.023                            & \textbf{0.338}  \\

\hline
\end{tabular*}
\end{table}

We may observe that, under faster mean-reversion, the average
volatility $\hat{\mathbb{E}}\bar{\sigma}^2_{m}$ and, consequently, the
price of the option are lower, which is exactly what is expected from
the model.

Tables~\ref{tabConv1} and \ref{tabConv2} illustrate how the price of
the option changes with the decrease of time step in discrete model.

% Table generated by Excel2LaTeX from sheet 'ax+b_dt'
%
%t3 ###
\begin{table}[t]
\centering
\caption{$\sigma^2(Y_s)=|Y_s|+0.2$, $K=1$, $r=0.02$, $k=0.1$, $T=1$.
Convergence}
\begin{tabular*}{\textwidth}{@{\extracolsep{\fill}}cD{.}{.}{3.0}cccc@{}}
\hline\\[-8pt]
$\Delta t$ & \multicolumn{1}{c}{$\alpha$} & $\hat{\mathbb{E}}\bar{\sigma}^2_{m}$ & $\bar{d}^{(m)}_1$ & $\bar{d}^{(m)}_2$ & $\hat{\mathbb{E}}\hat{P}_{m}$ \\[1pt]
\hline
$10^{-2}$ & 1 & 0.272367 & 0.299043 & $-$0.222056 & 0.213552 \\
$10^{-3}$ & 1 & 0.271043 & 0.298506 & $-$0.221338 & 0.213073 \\
$10^{-4}$ & 1 & 0.272534 & 0.299123 & $-$0.222179 & 0.213631 \\
$10^{-5}$ & 1 & 0.271837 & 0.298822 & $-$0.221753 & 0.213351 \\
$10^{-6}$ & 1 & 0.271421 & 0.298667 & $-$0.221560 & 0.213220 \\
$10^{-2}$ & 100 & 0.208910 & 0.272291 & $-$0.184776 & 0.189047 \\
$10^{-3}$ & 100 & 0.206599 & 0.271267 & $-$0.183264 & 0.188073 \\
$10^{-4}$ & 100 & 0.206439 & 0.271196 & $-$0.183159 & 0.188005 \\
$10^{-5}$ & 100 & 0.206413 & 0.271184 & $-$0.183142 & 0.187994 \\
$10^{-6}$ & 100 & 0.206443 & 0.271198 & $-$0.183162 & 0.188007 \\
\hline
\end{tabular*}
\label{tabConv1}%
\end{table}
%

% Table generated by Excel2LaTeX from sheet 'exp_dt'
%
%t4 ###
\begin{table}[t]
\centering
\caption{$\sigma^2(Y_s)=\operatorname{e}^{Y_s}+0.2$, $K=1$, $r=0.02$, $k=0.1$, $T=1$. Convergence}
\begin{tabular*}{\textwidth}{@{\extracolsep{\fill}}cD{.}{.}{3.0}cccc@{}}
\hline\\[-8pt]
$\Delta t$ & \multicolumn{1}{c}{$\alpha$} & $\hat{\mathbb{E}}\bar{\sigma}^2_{m}$ & $\bar{d}^{(m)}_1$ & $\bar{d}^{(m)}_2$ & $\hat{\mathbb{E}}\hat{P}_{m}$ \\[1pt]
\hline
$10^{-2}$ & 1 & 1.265414 & 0.556279 &   $-$0.519082 & 0.431865 \\
$10^{-3}$ & 1 & 1.269504 & 0.579243 &   $-$0.543620 & 0.432472 \\
$10^{-4}$ & 1 & 1.266274 & 0.584925 &   $-$0.549670 & 0.431990 \\
$10^{-5}$ & 1 & 1.266030 & 0.566934 &   $-$0.530485 & 0.431948 \\
$10^{-6}$ & 1 & 1.265635 & 0.576169 &   $-$0.540343 & 0.431892 \\
$10^{-2}$ & 100 & 1.201083 & 0.566092 & $-$0.529585 & 0.422128 \\
$10^{-3}$ & 100 & 1.201047 & 0.566500 & $-$0.530021 & 0.422123 \\
$10^{-4}$ & 100 & 1.201026 & 0.566203 & $-$0.529703 & 0.422120 \\
$10^{-5}$ & 100 & 1.201036 & 0.566693 & $-$0.530228 & 0.422121 \\
$10^{-6}$ & 100 & 1.201023 & 0.566052 & $-$0.529542 & 0.422119 \\
\hline
\end{tabular*}
\label{tabConv2}%
\end{table}

In view of Section~\ref{sec3}, it is also of certain interest to compare
calculations obtained over one trajectory but under different
discretization steps. We constructed 2000 trajectories with time-step
size of $10^{-6}$: 1000 for the case $\alpha=1$ and 1000 for the case
$\alpha=100$. These trajectories are considered to be ``true''
continuous-time trajectories of the Ornstein--Uhlenbeck process $Y_t$.
The corresponding values of $\bar{\sigma}^2_{m}$ are considered to be
``true'' continuous-time values of~$\bar{\sigma}^2$. The calculations
were then performed for wider discretization intervals using the points
of constructed trajectories. Thus, the samples of discretization errors
for $\bar{\sigma}^2_{m}$ were derived. Probably, the estimate of $\bar
{\sigma}^2_{m}$ is more valuable in such context since one would not
usually calculate the price of an option over one trajectory. However,
the estimate of volatility is usually derived from past data, which is
in essence one distinct realization of the space of all possible scenarios.

Tables~\ref{tabErr1} and \ref{tabErr2} provide characteristics of the
samples of discretization errors. Errors are measured as a percentage
of the ``true'' value.
% Table generated by Excel2LaTeX from sheet 'ax+b_dt_AvErr_Stats'
%
%t5 ###
\begin{table}[t!]
\tabcolsep=0pt
\caption{$\sigma^2(Y_s)=|Y_s|+0.2$, $K=1$, $r=0.02$, $k=0.1$, $T=1$.
Characteristics of sample of errors}\vspace*{-4pt}
\begin{tabular*}{\textwidth}{@{\extracolsep{\fill}}lllll@{}}
\hline\\[-8pt]
 & \multicolumn{1}{c}{$10^{-2}$} & \multicolumn{1}{c}{$10^{-3}$} & \multicolumn{1}{c}{$10^{-4}$} & \multicolumn{1}{c@{}}{$10^{-5}$} \\
\hline
& \multicolumn{4}{c}{$\alpha=1$} \\
\cline{2-5}
Average       & 0.08710\%                        & 0.00834\%  & 0.00081\%  & 0.00008\%  \\
St. error     & 0.0000427                        & 0.0000042  & 0.0000004  & 0          \\
Median        & 0.0009575                        & 0.0000834  & 0.000008   & 0.0000007  \\
St. deviation & 0.0013517                        & 0.0001334  & 0.0000137  & 0.0000013  \\
Excess        & $-$0.217306                        & $-$0.191189  & $-$0.143295  & $-$0.021156  \\
Skewness      & 0.0492335                        & $-$0.002248  & 0.023124   & 0.0577173  \\
Min           & $-$0.29706\%                       & $-$0.03669\% & $-$0.00303\% & $-$0.00036\% \\
Max           & 0.52352\%                        & 0.04766\%  & 0.00502\%  & 0.00044\%  \\
Count         & 1000                             & 1000       & 1000       & 1000       \\
\hline
              & \multicolumn{4}{c}{$\alpha=100$}                                        \\
\cline{2-5}
Average       & 0.07790\%                        & 0.00742\%  & 0.00083\%  & 0.00007\%  \\
St. error     & 0.000043                         & 0.0000044  & 0.0000004  & 0          \\
Median        & 0.0008379                        & 0.0000728  & 0.0000083  & 0.0000007  \\
St. deviation & 0.0013602                        & 0.0001379  & 0.0000136  & 0.0000014  \\
Excess        & $-$0.234452                        & $-$0.302723  & $-$0.352995  & $-$0.054568  \\
Skewness      & $-$0.024765                        & 0.0922374  & 0.0055451  & 0.0229423  \\
Min           & $-$0.30504\%                       & $-$0.03231\% & $-$0.00323\% & $-$0.00037\% \\
Max           & 0.46265\%                        & 0.04974\%  & 0.00454\%  & 0.00050\%  \\
Count         & 1000                             & 1000       & 1000       & 1000       \\

\hline
\end{tabular*}
\label{tabErr1}%
\end{table}
%

% Table generated by Excel2LaTeX from sheet 'exp_dt_AvErr_Stats'
%
%t6 ###
\begin{table}[t!]
\tabcolsep=0pt
\caption{$\sigma^2(Y_s)=\operatorname{e}^{Y_s}+0.2$, $K=1$, $r=0.02$,
$k=0.1$, $T=1$. Characteristics of sample of errors}\vspace*{-4pt}
\begin{tabular*}{\textwidth}{@{\extracolsep{\fill}}lllll@{}}
\hline\\[-8pt]
 & \multicolumn{1}{c}{$10^{-2}$} & \multicolumn{1}{c}{$10^{-3}$} & \multicolumn{1}{c}{$10^{-4}$} & \multicolumn{1}{c@{}}{$10^{-5}$} \\
\hline
              & \multicolumn{4}{c}{$\alpha=1$}                                          \\
\cline{2-5}
Average       & 0.02496\%                        & 0.00268\%  & 0.00026\%  & 0.00002\%  \\
St. error     & 0.0000113                        & 0.0000011  & 0.0000001  & 0.00000001 \\
Median        & 0.0002559                        & 0.0000266  & 0.0000027  & 0.0000002  \\
St. deviation & 0.0003584                        & 0.0000354  & 0.0000035  & 0.0000003  \\
Excess        & 0.1947561                        & 0.1687356  & $-$0.0576859 & 0.0700827  \\
Skewness      & $-$0.1691937                       & $-$0.0097185 & $-$0.1643507 & $-$0.0522007 \\
Min           & $-$0.09961\%                       & $-$0.00861\% & $-$0.00088\% & $-$0.00011\% \\
Max           & 0.12871\%                        & 0.01464\%  & 0.00126\%  & 0.00013\%  \\
Count         & 1000                             & 1000       & 1000       & 1000       \\
\hline
              & \multicolumn{4}{c}{$\alpha=100$}                                        \\
\cline{2-5}
Average       & 0.02692\%                        & 0.00268\%  & 0.00025\%  & 0.00002\%  \\
St. error     & 0.0000118                        & 0.0000012  & 0.0000001  & 0          \\
Median        & 0.0002712                        & 0.0000265  & 0.0000027  & 0.0000002  \\
St. deviation & 0.0003735                        & 0.0000377  & 0.0000036  & 0.0000003  \\
Excess        & 0.17242                          & 0.070383   & 0.3383414  & 0.0853763  \\
Skewness      & $-$0.0299531                       & $-$0.0195205 & $-$0.1914745 & $-$0.0371876 \\
Min           & $-$0.09174\%                       & $-$0.01068\% & $-$0.00112\% & $-$0.00011\% \\
Max           & 0.16291\%                        & 0.01411\%  & 0.00139\%  & 0.00014\%  \\
Count         & 1000                             & 1000       & 1000       & 1000       \\\hline
\end{tabular*}
\label{tabErr2}%
\end{table}
It can be seen from the tables that approximation results do not differ
significantly for various time-steps. Even the widest investigated
discretization interval provides acceptable precision for most applications.

%s5 ###
\section{Checking approximation precision in
the case of deterministic volatility}\label{sec5}

In this section, we compare the option prices obtained for the Euler
scheme \eqref{Yml} with the true prices of European call option for
different sets of parameters for the case of deterministic
time-dependent volatility.

The models with deterministic time-dependent volatility are the natural
extension of the Black--Scholes model. The expression for the price of
the option is the same as in the classical model except for the fact
that, instead of constant volatility, it operates with average (or root
mean square) volatility over the time interval to maturity (see, e.g.,
\cite{book9,book10}). Thus, the formula remains similar to \eqref
{IntExp} and~\eqref{V0discr}.\vadjust{\eject}

It has been shown that deterministic volatility does not reflect the
real-world stochastic dynamics correctly \cite{paper27,paper28}, and
such models have begun falling out of favor in the mid-1980s. The shift
to stochastic volatility models was boosted by rapid development of
computational tools.

Nevertheless, deterministic volatility is suitable for the purpose of
our investigation since we can calculate the exact price of the option
for the continuous time model.

In order to analyze the deterministic time-dependent volatility case,
it looks natural to let the Brownian noise term in the definition of
$Y_t$ vanish. Thus, we get
%
%e13 ###
\begin{equation}
\label{dYdeterm} dY_t = -\alpha Y_tdt,
\end{equation}
which is a familiar linear differential equation solved by
%
%e14 ###
\begin{equation}
\label{Ydeterm} Y_t = Y_0\operatorname{e}^{-\alpha t}.
\end{equation}
For the same transformation functions $\sigma$ and sets of parameters
as in the previous section, we calculate the prices of European call
option in the continuous case using \eqref{IntExp} and compare it with
the prices of the same option calculated using \eqref{V0discr}--\eqref
{sigmadiscr} with
%
%e15 ###
\begin{equation}
Y^{(m)}_{l+1}=(1-\alpha\Delta t)Y^{(m)}_{l}.
\end{equation}
We use the time step of 0.01 and only 10 simulations per combination of
inputs. As before, all calculations are performed in Matlab 7.9.0.

Table~\ref{tabDetervol} presents the results of calculations.
Comparison of two approaches reveals that the Euler--Maruyama scheme
provides a good approximation for the exact option price. In the case
of fast mean-reversion, the results coincide when rounded to sixth digit.

% Table generated by Excel2LaTeX from sheet 'ax+b_determ,1'
%
%t7 ###
\begin{table}[htbp]
\tabcolsep=0pt
\caption{Approximate option prices versus true option prices for deterministic volatility}
\begin{tabular*}{\textwidth}{@{\extracolsep{\fill}}D{.}{.}{1.2}D{.}{.}{1.1}D{.}{.}{1.2}D{.}{.}{1.1}cD{.}{.}{1.1}cccc@{}}
\hline\\[-8pt]
\multicolumn{1}{c}{$T$} & \multicolumn{1}{c}{$\alpha$} & \multicolumn{1}{c}{$r$} & \multicolumn{1}{c}{$K$} & \multicolumn{1}{c}{$a$} & \multicolumn{1}{c}{$b$} & $\hat{\mathbb{E}}\hat{P}_{m}$ & $\mathbb{E}V$ & $\hat{\mathbb{E}}\hat{P}_{m}$ & $\mathbb{E}V$ \\
\hline
& & & & & & \multicolumn{2}{c}{$\sigma^2(Y_s)=a\|Y_s\|+b$} & \multicolumn{2}{c}{$\sigma^2(Y_s)=\operatorname{e}^{Y_s}+0.2$} \\
\cline{7-8}\cline{9-10}
0.25 & 1   & 0    & 0.8 & 1 & 0   & 0.203891 & \textbf{0.203888} & 0.316223 & \textbf{0.316220} \\
0.5  & 1   & 0    & 0.8 & 1 & 0   & 0.211556 & \textbf{0.211549} & 0.390150 & \textbf{0.390147} \\
1    & 1   & 0    & 0.8 & 1 & 0   & 0.223003 & \textbf{0.222994} & 0.490305 & \textbf{0.490302} \\
0.25 & 1   & 0.01 & 1   & 1 & 0.2 & 0.107942 & \textbf{0.107935} & 0.224736 & \textbf{0.224733} \\
0.5  & 1   & 0.01 & 1   & 1 & 0.2 & 0.150207 & \textbf{0.150199} & 0.312794 & \textbf{0.312791} \\
1    & 1   & 0.01 & 1   & 1 & 0.2 & 0.206464 & \textbf{0.206457} & 0.429067 & \textbf{0.429064} \\
0.25 & 1   & 0.02 & 1.2 & 1 & 1   & 0.141317 & \textbf{0.141313} & 0.159958 & \textbf{0.159954} \\
0.5  & 1   & 0.02 & 1.2 & 1 & 1   & 0.227633 & \textbf{0.227629} & 0.253710 & \textbf{0.253706} \\
1    & 1   & 0.02 & 1.2 & 1 & 1   & 0.345261 & \textbf{0.345257} & 0.379955 & \textbf{0.379952} \\
0.25 & 100 & 0    & 0.8 & 1 & 0   & 0.200000 & \textbf{0.200000} & 0.309950 & \textbf{0.309950} \\
0.5  & 100 & 0    & 0.8 & 1 & 0   & 0.200000 & \textbf{0.200000} & 0.382107 & \textbf{0.382106} \\
1    & 100 & 0    & 0.8 & 1 & 0   & 0.200000 & \textbf{0.200000} & 0.481610 & \textbf{0.481610} \\
0.25 & 100 & 0.01 & 1   & 1 & 0.2 & 0.091044 & \textbf{0.091044} & 0.217149 & \textbf{0.217149} \\
0.5  & 100 & 0.01 & 1   & 1 & 0.2 & 0.128449 & \textbf{0.128449} & 0.303457 & \textbf{0.303457} \\
1    & 100 & 0.01 & 1   & 1 & 0.2 & 0.181507 & \textbf{0.181507} & 0.419198 & \textbf{0.419198} \\
0.25 & 100 & 0.02 & 1.2 & 1 & 1   & 0.133108 & \textbf{0.133108} & 0.152065 & \textbf{0.152065} \\
0.5  & 100 & 0.02 & 1.2 & 1 & 1   & 0.217100 & \textbf{0.217100} & 0.243748 & \textbf{0.243748} \\
1    & 100 & 0.02 & 1.2 & 1 & 1   & 0.333759 & \textbf{0.333759} & 0.369312 & \textbf{0.369312} \\

\hline
\end{tabular*}
\label{tabDetervol}%
\end{table}

\begin{remark}
In this paper, we consider the price of the option at the initial time
moment. However, all the above considerations are applicable for any
valuation date $t$ between the initial time moment and maturity. Some
obvious changes need to be made, for example, the function
$\bar{\sigma}_t:=\sqrt{\frac{1}{T-t}\int_t^T\sigma^2(Y_s)ds} \geq0$ needs to be
introduced instead of $\bar{\sigma}$, and $T$ needs to be substituted
by $T-t$ in \eqref{IntExp}--\eqref{sigmadiscr}.
\end{remark}

\appendix

\section*{Appendix A. The Euler scheme: definitions and auxiliary results}

The reader is advised to refer to \cite{book6}, which provides in-depth
study of numerical approximations of stochastic differential equations.

Consider the stochastic differential equation
%
%e16 ###
\begin{eqnarray}
\label{GeneralIto} dX_t=a(t,X_t)dt+b(t,X_t)dW_t,
\quad t \in[t_0, T],
\end{eqnarray}
and assume that there is a unique strong solution $X(t)$ with
$X(t_0)=X_0$. In order for this to be the case, certain assumptions
need to be made about the functions $a$ and $b$. Namely, refer to the
following assumptions (assumptions (A1)--(A4) in \cite{book6}, pp.~128--129):
\begin{itemize}
\item[A1)] $a=a(t,x)$ and $b=b(t,x)$ are jointly $L^2$-measurable in
$(t,x) \in[t_0,T] \times\mathds{R}$;
\item[A2)] the functions $a$ and $b$ satisfy the Lipschitz condition
w.r.t. $x$, that is, there exists a constant $K>0$ such that
\[
\big|a(t,x)-a(t,y)\big|\leq K|x-y|
\]
and
\[
\big|b(t,x)-b(t,y)\big|\leq K|x-y|
\]
for all $t\in[t_0,T]$ and $x,y \in\mathds{R}$;
\item[A3)] there exists a constant $K>0$ such that
\[
\big|a(t,x)\big|^2\leq K\big|1+|x|^2\big|
\]
and
\[
\big|b(t,x)\big|^2\leq K\big|1+|x|^2\big|
\]
for all $t\in[t_0,T]$ and $x,y \in\mathds{R}$;
\item[A4)] $X_{t_0}$ is $\mathcal{F}_{t_0}$-measurable with $\mathbb
{E}|X_{t_0}|^2<\infty$.
\end{itemize}
Let $X^{(m)}_t$ be a discretization scheme of the process $X_t$.
\begin{ndefin}(See \cite{book6}.)
We shall say that an approximating process $X^{(m)}_t$ converges in the
strong sense with order $\gamma\in(0, \infty]$ to the true process
$X_t$ if there exists a finite constant $K$ such that
\begin{equation*}
\mathbb{E} \bigl( \big|X_t-X^{(m)}_t\big| \bigr)\le K
m^{-\gamma}.
\end{equation*}
\end{ndefin}

The same terminology will be applied to the functions of approximating
processes.

\begin{ndefin}(See \cite{book6}.)
We shall say that a discrete time approximation scheme $X^{(m)}_t$ is
strongly consistent if there exists a nonnegative function $c=c(m)$ with
\[
\lim_{m \to\infty}c(m)=0
\]
such that
\begin{equation*}
\mathbb{E} \biggl( \biggl\llvert\mathbb{E} \biggl(\frac
{X^{(m)}_{i+1}-X^{(m)}_{i}}{T/m } \Bigm|
\mathcal{F}_{iT/m} \biggr) -a \biggl(\frac{iT}{m},X^{(m)}_i
\biggr) \biggr\rrvert^2 \biggr)\leq c(m)
\end{equation*}
and
\begin{equation*}
\mathbb{E} \biggl( \frac{m}{T} \bigg|X^{(m)}_{i+1}-X^{(m)}_i-
\mathbb{E} \bigl(X^{(m)}_{i+1}-X^{(m)}_i\big|
\mathcal{F}_{iT/m} \bigr)
\\
-b \biggl(\frac{iT}{m},X^{(m)}_i \biggr)\Delta
W_i\bigg|^2 \biggr) \leq c(m)
\end{equation*}
for all fixed values $X^{(m)}_i=y$ and $i=0,1, \dots,m$.
\end{ndefin}
%
%Matters of convergence of discrete approximations of stochastic
%processes are investigated, for example, in \cite{book6}, ...

\begin{nthm}(\xch{See \cite{book6}}{cite{book6}}, 9.6.2, p. 324.)
Let assumptions (A1)--(A4) hold for \eqref{GeneralIto}. Then a strongly
consistent equidistant-time discrete approximation $X^{(m)}$ of the
process $X$ on $[t_0,T]$, with $X^{(m)}_{t_0}=X_{t_0}$, converges
strongly to $X$.
\end{nthm}

Evidently, the Euler scheme $Y^{(m)}$ introduced to approximate $Y$ in
Section~~\ref{sec2} satisfies all the above requirements and hence converges
strongly. Moreover, it is a well-known fact that, in general, the
convergence of the Euler approximation is of order 0.5. One can check
these propositions using the estimates of the rate of convergence
provided in \cite{book6} by the proof of Theorem 9.6.2 and Exercise 9.6.3.

However, our case is more specific since $Y^{(m)}$ approximates the
diffusion process with additive noise, that is, $b(t,x)=k$ is
constant.Hence, the following proposition holds.

\begin{ntve}
\label{Conv1}
$Y^{(m)}$ is the Milstein scheme and thus converges strongly with order 1.
\end{ntve}

Really, the only difference in representation of $Y^{(m)}$ as the
Milstein scheme compared to the Euler one is in the additional summand
of the form
\[
\frac{1}{2}bb'\bigl(\bigl(\Delta Z^\mathbb{Q}
\bigr)^2-T/m\bigr),
\]
which is identically zero for the constant function $b$. The Milstein
scheme is known to converge with order 1 (see, e.g., \cite{book6},
Theorem 10.6.3, p.~361).

%\bibliography{bib/biblio}
%

%
\end{document}